\documentclass[10pt,reqno]{amsart}
\usepackage{graphicx}
\usepackage{indentfirst,csquotes,biblatex}
\usepackage{amssymb,amsthm,amsmath}
\usepackage{xcolor,paralist,hyperref,titlesec,fancyhdr,etoolbox}
\newtheorem{theorem}{Theorem}[]

\titleformat{\section}{\normalfont\Large\bfseries}{\thesection.}{0.5em}{}
\titleformat{\subsection}{\normalfont\centering\large\bfseries}{\Alph{subsection}. }{0em}{}
\titlespacing*{\section}{0pt}{*4}{*1}

\hypersetup{ colorlinks=true, linkcolor=black, filecolor=black, urlcolor=black }

\usepackage{lipsum}

\begin{document}
\title{A Thermodynamically Universal Turing Machine}
\author{Jihai Zhu}
\date{\today}

\email{raintown2048@163.com}
\maketitle

\begin{abstract}
     Expanding upon the widely recognized notion of mathematical universality in Turing machines, a concept of thermodynamic universality in Turing machines is introduced. Under the physical Church-Turing thesis, the existence of a thermodynamically universal Turing machine (TUTM) is demonstrated. A TUTM not only has the capability to simulate the input-output behavior of any given Turing machine but also replicate the heat production of that machine up to an additive constant. The finding shows that the hypothesis that the physical world is simulated by Turing machines may not be completely absurd.
\end{abstract} 

\section{Introduction}
    Turing machines (TMs) possess a crucial property known as (mathematical) universality, which means that there exists a class of TMs called \textit{universal Turing machines} (UTMs) capable of simulating any given TM with a specified input program. However, in the physical world, every TM (or any computational device, such as a laptop, that possesses the same computational power as a UTM), has its own physical properties, such as heat generation and work dissipation during operation. This leads us to inquire whether there exists a TM that is \textit{physically universal} in the sense that it can not only simulate any other TM mathematically but also replicate certain physical properties of the specified TM. (Note that the concept of "physical universality" discussed in this paper should not be confused with the concept found in papers on cellular automata\cite{saloPhysicallyUniversalTuring2020}.) 

    When discussing physical universality, it may be better to restrict the scope of physical properties. This paper specifically focuses on the thermodynamic properties of TMs, building upon the research that has been conducted since the 1960s regarding the relationship between thermodynamics and computation. Early studies on the thermodynamics of Turing machines primarily concentrated on logically reversible TMs\cite{bennettLogicalReversibilityComputation1973,bennettTimeSpaceTradeOffs1989}. However, more recently, the work by Kolchinski and Wolpert\cite{kolchinskyThermodynamicCostsTuring2020} has broadened the scope to encompass general TMs, which serves as the main foundation for this paper. In Kolchinski and Wolpert's paper, a TM in the physical world is identified as a discrete-state physical system that is coupled to a heat bath at temperature $T$ and that evolves under the influence of a driving protocol, and the inputs and outputs of the TM are identified as the initial and final states of the physical system, so that the computation performed by the TM is simulated by the dynamics over the status of the physical system. Following Kolchinski and Wolpert's notation, a physical process that adheres to the physical laws and exhibits dynamics corresponding to the input-output map of a TM is referred to as a \textit{realization} of that TM.   

    With the above identification, techniques from stochastic thermodynamics\cite{vandenbroeckEnsembleTrajectoryThermodynamics2015} can be applied to the analysis of the thermodynamics properties of realizations of TMs. In this paper, we particularly focus on the amount of heat, denoted as $Q(x)$, generated when a TM $M$ is executed on each input $x$ and gives the output $M(x)$ after a finite number of steps. We refer to $Q(x)$ as the \textit{heat function} of $M$, and the pair $(M,Q)$ as a \textit{thermodynamic Turing machine}(TTM). 
    
    The pair $(M,Q)$ must obey some condition in order that there exists a realization of $M$ with $Q(x)$ as its heat function, or that the TTM $(M,Q)$ is realizable. The complete condition for the realizability of a TTM has been derived in\cite{kolchinskyThermodynamicCostsTuring2020}, as is recapitulated in Section 3.
    
    The main result of this paper is that, under the assumption of \textit{physical Church-Turing thesis}(PCTT), there exists a \textit{thermodynamically universal Turing machine}(TUTM) which is a (limiting) realizable TTM $(U,Q_U)$ and that, given any realizable TTM $(M,Q)$, it can not only simulate the input-output map $M$, but also replicate the heat function $Q$ (up to an additive constant determined by $M$ and $Q$), with a specified input program.    

    The paper is structured as follows: Section 2 provides a review of the basic concepts of TMs and \textit{algorithmic information theory} (AIT), which will be utilized later in the paper. Section 3 and Section 4 recapitulate the results derived in \cite{kolchinskyThermodynamicCostsTuring2020} regarding the realizability conditions of a TTM and the dominating realizations of TMs under PCTT. Finally, in Section 5, the existence of the TUTM under PCTT is demonstrated.

\section{Background of TMs and AIT}
    \subsection{Turing machines}
        A TM is an abstract model of computation consisting 3 variables:
        \begin{enumerate}
            \item a \textit{tape} variable which is a semi-infinite string $s\in\{0,1,b\}^\infty$, where 0 and 1 are the binary symbols and $b$ is the blank symbol; 
            \item a \textit{pointer} variable $v\in\mathbb{N}$, which is interpreted as specifying a “position” on the tape;
            \item a \textit{head} variable $q$ belongs to a finite set $Q$ of states, which includes a specially designated \textit{start state} and a specially designated \textit{halt state}.
        \end{enumerate} 

        The TM starts with its head in the start state, the pointer set to position 0, and its tape containing some finite string of non-blank symbols, followed by blank symbols. The joint state of the tape, pointer, and head evolves over time according to a discrete-time \textit{update function}, which must be bounded in a finite region. If during that evolution the head ever enters its halt state, the computation considered completed and the contents of the tape at that time serves as the output. However, for certain inputs, a TM may never reach the halt state, thereby failing to complete the computation.

        The input-output map of a TM $M$ can be associated with a (partial) function $M: \mathbb{N} \rightarrow \mathbb{N}$ when a one-to-one correspondence between $\mathbb{N}$ and the set of finite strings $\{0,1\}^*$ is established. 

        A significant characteristic of TMs is their (mathematical) universality. This means that there exists a TM $U: \mathbb{N} \times \mathbb{N} \rightarrow \mathbb{N}$ which, given any TM $M$, there is an index $\ulcorner M\lrcorner \in \mathbb{N}$ such that for every input $x$, $U(\ulcorner M\lrcorner, x) = M(x)$. Essentially, TM $U$ can simulate the computation carried out by any TM and is referred to as a UTM.

    \subsection{Algorithimic information theory}
        AIT studies complexity measures on strings. Because most objects with interest can be described in terms of strings, or as the limit of a sequence of strings, AIT can be applied to lots of areas.

        The \textit{algorithmic complexity} of a string $x$ is defined as the length of the shortest program that, when fed into a UTM $U$, produces $x$ as output:
        \begin{equation}
            C_U(x)\equiv\min_{p:U(p)=x}l(p)
        \end{equation}
        Because any two UTMs $U$ and $U'$ can simulate each other with programs of finite length, $|C_U(x)-C_{U'}(x)|\le c_{U,U'}$. As the algorithmic complexity of a string approaches infinity as its length increase, an additive constant may be ignored and we simply write $C(x)$ for $C_U(x)$.

        Similarly, The \textit{conditional algorithmic complexity} of $x\in\{0,1\}^*$ given $y\in\{0,1\}^*$ is defined as the length of the shortest program that, when paired with $y$ and then fed into a UTM $U$, produces $x$ as output:
        \begin{equation}
            C_U(x|y)\equiv\min_{p:U(y,p)=x}l(p)
        \end{equation}
        For the same reason, we may omit the subscript $U$ and write $C(x|y)$ for $C_U(x|y)$.

        While the algorithmic complexity of a string well characterizes the intuition complexity of a string, it lacks some good properties we want a complexity concept to observe, such as the semi-additivity\cite{liKolmogorovComplexityIts1990}:
        \begin{equation}
            C(x,y)\nleq C(x)+C(y)
        \end{equation}
        To derive a concept of complexity which has better mathematical properties, we have to first turn to the concept of \textit{prefix sets}:A set of strings $A\subset\{0,1\}^*$ is called prefix if there is no two strings in $A$ such that one is the prefix of the other, formally:
        \begin{equation}
            \forall x,y\in A\forall z\in\{0,1\}^*(xz\neq y)
        \end{equation}
        Every prefix set satisfies the \textit{Kraft's inequality}:
        \begin{equation}
            \sum_{x\in A}2^{-l(x)}\leq 1
        \end{equation}

        Now it is able to define the concept of \textit{prefix Turing machines}: a prefix TM is a TM whose domain is a prefix set. There also exists prefix UTMs, so we can define the \textit{prefix algorithmic complexity} of $x\in\{0,1\}^*$ as:
        \begin{equation}
            K_U(x)\equiv\min_{p:U(p)=x}l(p)
        \end{equation}
        where $U$ is a prefix UTM.

        The \textit{prefix conditional algorithmic complexity} $K_U(x|y)$ can be similarly defined and we may write $K(x)$, $K(x|y)$ for $K_U(x)$, $K_U(x|y)$.

\section{A Theorem on the Realizability of a TTM}
    We consider a physical system with a countable state space $X$, which is connected to a heat bath at temperature $T$ and a work reservoir. The system evolves under the influence of a driving protocol, and we are interested in its dynamics over some fixed interval $t\in[0,1]$. In this paper, only deterministic dynamics is considered, so every dynamics on $X$ can be associated to a (partial) function $f:X\to X$, and we define $Q(x)$ as the heat transferred to the heat bath during $[0,1]$. A pair $(f,Q)$ of a (partial) function $f:X\to X$ together with a (partial) function $Q:X\to \mathbb{R}$ is called \textit{realizable} if there exists a physical system with $f$ as its dynamics and $Q$ as its heat function.  
    
    To simplify matters, we assume $k_BT=1$ henceforth, allowing us to equate heat transfer into (or out of) the heat bath with the entropy flow into (or out of) the heat bath. 

    As is the common practice in stochastic thermodynamics, we denote $p_X$ as the initial distribution of the system over its state space, and $p_{f(X)}$ as the final distribution.
    
    The following technical result links the logical properties of a (partial) function f with the heat function of any realization of that $f$ . This result will be central to our analysis, as it will allows us to establish thermodynamic constraints on processes that realize TMs.
    \begin{theorem}
        Given a countable set $X$, let $f:X\to X$ and $Q:X\to\mathbb{R}$ be two (partial) functions with the same domain of definition. The following are equivalent:
        \begin{enumerate}
            \item For all $p_X$, 
                \begin{equation}
                    \langle Q\rangle_{p_X}+S(p_{f(X)})-S(p_X) \ge 0
                \end{equation}
            \item For all $y\in\operatorname{im}f$, 
                \begin{equation}
                    \sum_{x:f(x)=y}2^{-Q(x)} \le 1
                    \label{9}
                \end{equation}
            \item (f,Q) is realizable.
        \end{enumerate}
    \end{theorem}
    \begin{proof}
        See\cite{kolchinskyThermodynamicCostsTuring2020}.
    \end{proof}

\section{The Dominating Realization of a TM}
    The physical Church-Turing thesis(PCTT) states that any function that can be computed by a physical system can be computed by a TM\cite{piccininiComputationPhysicalSystems2021}. In the subsequent sections of this paper, we adopt the assumption that PCTT holds. This assumption requires that the heat function $Q(x)$ of every realizable TTM must be computable. 

    \cite{kolchinskyThermodynamicCostsTuring2020} demonstrated that, for every TM $M$, under PCTT, there exists a (limiting) realizable TTM $(M,Q_{dom})$, where $Q_{dom}(x)\equiv K(x|M(x))$ is smaller than any other realizable heat function $Q$, up to an additive constant determined by $M$ and $Q$. Formally:
    \begin{theorem}
        The heat function $Q_{dom}(x)\equiv K(x|M(x))$ is physically realizable and is better than any other realization of $M$ with heat function $Q$, in the sense that
        \begin{equation}
            Q_{dom}(x)\le Q(x)+K(Q)+K(M)+O(1)
        \end{equation}
    \end{theorem}
    \begin{proof}
        See\cite{kolchinskyThermodynamicCostsTuring2020}.
    \end{proof}

    It should be noted that $K(x|M(x))$ is normally uncomputable, but rather upper semi-computable, which means that there is a decreasing sequence $Q_0\geq Q_1\geq\cdots$ whose limit $\lim_{n\to\infty}Q_n=Q_{dom}$. Thus we can consider the dominating heat function as the limit of a sequence of realizable heat functions. Generally, we say a TTM $(M,Q)$ is \textit{limiting realizable} if there exists decreasing realizable heat functions of $M$: $Q_0\geq Q_1\geq\cdots$, whose limit $\lim_{n\to\infty}Q_n=Q$.
    
\section{The Existence of TUTM}
    We now demonstrate the main result of this paper, that there exists a UTUM $(U,Q_U)$, which can simulate both the input-output map and the heat function of any realizable TTM $(M,Q)$.
    
    Before the formal statement and proof of the result, it is better to grasp some intuition of it. As a UTM has an input parameter $\ulcorner M\lrcorner$ to simulate the input-output map of any given TM $M$, a UTUM should have two input parameters, $\ulcorner M\lrcorner$ and $\ulcorner Q\lrcorner$, so as to simulate both the input-output map and the heat function of any realizable TTM $(M,Q)$.
    
    \begin{theorem}
        There exists a TUTM $U$ with heat function $Q_U$ that is limiting realizable and that for any other realizable TTM $(M,Q)$:
        
                \begin{equation}
                    U(\ulcorner M\lrcorner, \ulcorner Q\lrcorner, x)=M(x)
                    \label{11}
                \end{equation}
        
                \begin{equation}
                    |Q_U(\ulcorner M\lrcorner, \ulcorner Q\lrcorner, x)-Q(x)|\le K(M)+K(Q)+O(1)
                    \label{12}
                \end{equation}
        
    \end{theorem}
    \begin{proof}
        The existence of a UTM which has an invalid second parameter as in Eq.\ref{11} is obvious. We only have to define a limiting realizable heat function $Q_U$ for $U$ to satisfy Eq.\ref{12}.

        Define
        \begin{equation}
            Q_U(\ulcorner M\lrcorner, \ulcorner Q\lrcorner, x)\equiv Q(x)+K(M)+K(Q)
        \end{equation}
        $Q_U$ is upper semi-computable, so we only have to prove that $Q_U$ satisfies Eq.\ref{9}. As the definition of the prefix algorthmic complexity include a prefix set, it satisfies the Kraft's inequality:
        \begin{equation}
            \sum_x2^{-K(x)}\leq 1
        \end{equation}
        So the following inequality holds, thus completing the proof (the last inequality comes from the realizability of $Q$):
        \begin{equation}
            \sum_{e,i,x:M_e(x)=y}2^{-Q_U(e,i,x)}\leq\sum_{e,i,x}2^{-Q(x)}2^{-K(e)}2^{-K(i)}\leq\sum_{x}2^{-Q(x)}\leq 1
        \end{equation}
    \end{proof}

\section{Discussion}
    We have demonstrated the existence of TUTM in the last section, which tells us that, a physically universal TM may in some sense exist and with that machine, we may simluate the physical world with only numeric input. It shows that the hypothesis that the physical world is simulated by Turing machines may not be completely absurd.

\bigskip
\printbibliography
\end{document}